\pdfoutput=1
\documentclass{article}
\RequirePackage{amsthm,amsmath,amsfonts,amssymb}
\RequirePackage[numbers]{natbib}
\RequirePackage{tikz}
\RequirePackage[colorlinks,citecolor=blue,urlcolor=blue]{hyperref}
\RequirePackage{graphicx}
\usepackage[margin=1in]{geometry}
\usepackage{tikz-cd}
\usepackage{doi}

\usepackage{pgfplots}
\pgfplotsset{compat=1.15}
\usepackage{mathrsfs}
\usetikzlibrary{arrows}
\usetikzlibrary{arrows.meta}

\theoremstyle{plain}

\newtheorem{theorem}{Theorem}[section]
\newtheorem{lemma}[theorem]{Lemma}
\newtheorem{proposition}[theorem]{Proposition}

\theoremstyle{definition}
\newtheorem{definition}[theorem]{Definition}

\newtheorem*{remark}{Remark}

\title{Markov chain Monte Carlo Significance tests}
\author{Michael Howes\footnote{Correspondence email \href{mailto:mhowes@stanford.edu}{mhowes@stanford.edu}}\\ Department of Statistics, Stanford University}
\date{June 18, 2023}
\begin{document}
\maketitle

\begin{abstract}
    Monte Carlo significance tests are a general tool that produce p-values by generating samples from the null distribution. However, Monte Carlo tests are limited to null hypothesis which we can exactly sample from.  Markov chain Monte Carlo (MCMC) significance tests are a way to produce statistical valid p-values for null hypothesis we can only approximately sample from. These methods were first introduced by Besag and Clifford in 1989 and make no assumptions on the mixing time of the MCMC procedure. Here we review the two methods of Besag and Clifford and introduce a new method that unifies the existing procedures. We use simple examples to highlight the difference between MCMC significance tests and standard Monte Carlo tests based on exact sampling. We also survey a range of contemporary applications in the literature including goodness-of-fit testing for the Rasch model, tests for detecting gerrymandering \cite{Gerry1} and a permutation based test of conditional independence \cite{CPT}.
\end{abstract}

\section{Introduction}\label{intro}

Suppose we have observed a data point $X_0$ belonging to some space $\mathcal{X}$. A simple statistical hypothesis asks whether $X_0$ was drawn from a fixed distribution $\pi$ on $\mathcal{X}$. One way to test this hypothesis is by a Monte Carlo simulation. We generate a ``sea'' of hypothetical data points $\widetilde{X}_1,\ldots,\widetilde{X}_M$ each sampled independently according to the distribution $\pi$. The observed data $X_0$ is  compared to the samples $\widetilde{X}_1,\ldots,\widetilde{X}_M$ which together approximate $\pi$. In the applications considered here, the space $\mathcal{X}$ will be high dimensional. Thus, comparing $X_0$ to $\widetilde{X}_1,\ldots,\widetilde{X}_M$ requires choosing a test statistic $T$ that maps $X_0$ to some real-valued quantity of interest $T(X_0)$. If $T(X_0)$ is larger than most of $T(\widetilde{X}_1),\ldots,T(\widetilde{X}_M)$, then we have evidence against the null hypothesis $X_0 \sim \pi$. 

For certain distributions $\pi$, it may be difficult or impossible to create independent and identically distributed samples $\widetilde{X}_1,\ldots,\widetilde{X}_M \stackrel{\text{iid}}{\sim} \pi$.  For example, $\pi$ may be the uniform distribution on a large combinatorial set whose elements cannot be enumerated \cite{BC,Gerry1,Gerry2,algebraic}. Or, the null distribution $\pi$ may only be known up to a normalizing constant \cite{GoF,CPT,yash}. In \cite{BC}, Besag and Clifford demonstrated how to use Markov chain Monte Carlo (MCMC) techniques to sample $\widetilde{X}_1,\ldots,\widetilde{X}_M$ and keep the statistical validity of the Monte Carlo test. Their methods place minimal assumptions on the Markov chain and only require that $\pi$ is a stationary distribution of the Markov chain. The chain need not be rapidly mixing or even connected.

Here we review the theory of MCMC significance tests while emphasizing the role of exchangeability. Important properties of these methods are studied through examples, and we end with a survey of applications.

\subsection{Overview}

The paper is structured as follows. In Section~\ref{sec:MC tests} we review Monte Carlo tests and discuss the validity and consistency of Monte Carlo tests. In Section~\ref{sec:E MCMC} we show how to use MCMC to conduct valid significance tests. We describe the well-known MCMC tests introduced in \cite{BC} and a new method. We compare these methods and discuss their relative advantages. Two examples in Section~\ref{sec:examples} give insight into how MCMC tests work, and how they differ from standard Monte Carlo tests. Finally, in Section~\ref{sec:applications} we review applications of MCMC tests to goodness-of-fit testing \cite{GoF,algebraic}, the detection of gerrymandering \cite{Gerry1,Gerry2} and conditional independence testing \cite{CPT}. 

\subsection{Notation and background}

Throughout this paper $\mathcal{X}$ will be a set and $X_0 \in \mathcal{X}$ will represent observed data. The focus is on testing a hypothesis of the form $X_0 \sim \pi$ where $\pi$ is a fixed distribution on $\mathcal{X}$. 

A Markov chain on $\mathcal{X}$ will be described by a \emph{transition kernel} $K$. That is, a sequence $X_0,X_1,\ldots,X_M$ will be a Markov chain with transition kernel $K$ if for all $i=1,\ldots,M$
\[X_i \mid X_0,X_1,\ldots,X_{i-1} \sim K(X_{i-1},\cdot), \]
where $K(x,\cdot)$ is a distribution on $\mathcal{X}$ that depends on $x$. We will call the sequence $X_0,X_1,\ldots,X_M$ a Markov chain generated by $K$ or a Markov chain that transitions according to $K$.

We will always assume that the null distribution $\pi$ is a stationary distribution of $K$. That is, the distribution $\pi$ is preserved by the transition kernel $K$. Formally, if $X_0 \sim \pi$ and $X_1 \mid X_0 \sim K(X_0,\cdot)$, then marginally $X_1 \sim \pi$ also. As discussed above, we do not assume that $\pi$ is the \emph{only} stationary distribution of $K$. 

For simplicity, we will assume that $\pi$ has a density $f:\mathcal{X} \to [0,\infty)$ with respect to some reference measure $\lambda$. We will also assume that $K$ has a transition kernel density $k:\mathcal{X} \times \mathcal{X} \to [0,\infty)$ with respect to $\lambda$. These assumptions mean that
\begin{align*}
    \pi(S) &= \int_S f(x) \lambda(dx),
\end{align*}
and
\begin{align*}
    K(x_0,S) &= \int_S k(x_0,x)\lambda(dx),
\end{align*}
for all $x_0 \in \mathcal{X}$ and (measurable) subsets $S \subseteq \mathcal{X}$. Under these assumptions, the condition that $\pi$ is a stationary distribution for $K$ is equivalent to
\[
    f(y) = \int_\mathcal{X} f(x)K(x,y)\lambda(dx) \text{ for all } y \in \mathcal{X}. 
\]
The $L$-th power of the transition kernel $K$ will be written as $K^L$. Specifically, if $X_0,X_1,\ldots,X_{ML}$ is a Markov chain generated by $K$, then $X_0,X_L,X_{2L}\ldots,X_{ML}$ is a Markov chain generated by $K^L$. In words, $L$ steps of a Markov chain with transition kernel $K$ corresponds to one step of a Markov chain with transition kernel $K^L$. Since $\pi$ is a stationary distribution of $K$, $\pi$ is also a stationary distribution for $K^L$.

\section{Monte Carlo tests}\label{sec:MC tests}

Monte Carlo methods are used in statistical testing when the distribution of a test statistic under the null cannot be computed analytically. A Monte Carlo significance test of the hypothesis $H:X_0 \sim \pi$ proceeds as follows:
\begin{enumerate}
    \item Generate samples $\widetilde{X}_1,\ldots,\widetilde{X}_M \stackrel{\text{iid}}{\sim} \pi$ independently of $X_0$.
    \item Calculate the test statistic $T_0=T(X_0) \in \mathbb{R}$ and statistics $\widetilde{T}_1 = T(\widetilde{X}_1),\ldots, \widetilde{T}_M = T(\widetilde{X}_M)$ for comparison.
    \item Return 
    \begin{align}p_{MC} &= \frac{|\{1 \le i \le M : \widetilde{T}_i \ge T_0\}| + 1}{M+1}\label{eq p-value}\\
    &=\frac{\sum_{i=1}^M I[\widetilde{T}_i \ge T_0] + 1}{M+1},\nonumber \end{align}
    where $|S|$ denotes the size of a set $S$, and $I[S]$ is the indicator random variable for the event $S$.
\end{enumerate}
The quantity $p_{MC}$ is a \emph{Monte Carlo p-value}. If we wish to emphasize the fact that  $\widetilde{X}_1,\ldots,\widetilde{X}_M$ are i.i.d., then we will call $p_{MC}$ the \emph{standard} Monte Carlo p-value. Later we will discuss using MCMC to sample $\widetilde{X}_1,\ldots,\widetilde{X}_M$. This motivates the following definition.
\begin{definition}
    A \emph{sampler} is a randomized algorithm that takes as input an initial point $X_0 \in \mathcal{X}$ and a required number of samples $M$. The output is a random sequence of samples $\widetilde{X}_1,\ldots,\widetilde{X}_M \sim P_{X_0}$.
\end{definition}
Given a sampler, we can sample $\widetilde{X}_1,\ldots,\widetilde{X}_M \sim P_{X_0}$  and then calculate $p_{MC}$ as in \eqref{eq p-value}. The properties of $p_{MC}$ will depend on the sampler. Two properties, validity and consistency, are of particular importance.

\subsection{Validity}\label{sec:validity}

Small values of $p_{MC}$ give statistical evidence against the null hypothesis $H : X_0 \sim \pi$. If we observe $p_{MC}\le \alpha$, then we can reject $H$ at significance level $\alpha$.  For this to be a valid test, we require a bound on the probability of falsely rejecting $H$. That is, when the null hypothesis holds, the probability of observing $p_{MC}\le \alpha$ should be at most $\alpha$. If this holds for all possible values of $\alpha \in [0,1]$, then we will say that the sampler used to produce $\widetilde{X}_1,\ldots,\widetilde{X}_M$ is \emph{valid}.

The validity of a sampler is closely related to the exchangeability of $X_0,\widetilde{X}_1,\ldots,\widetilde{X}_M$. In particular, we are interested in samplers that make $X_0,\widetilde{X}_1,\ldots,\widetilde{X}_M$ exchangeable under the null $X_0 \sim \pi$.
\begin{definition}\label{def:exchangeable}
    A sampler $P_{X_0}$ is an \textit{exchangeable sampler} if $X_0 \sim \pi$ and $\widetilde{X}_1,\ldots,\widetilde{X}_M \mid X_0 \sim P_{X_0}$ implies that $X_0,\widetilde{X}_1,\ldots,\widetilde{X}_M$ are exchangeable. That is, if the null hypothesis holds and $\widetilde{X}_1,\ldots,\widetilde{X}_M$ are sampled according to $P_{X_0}$, then the joint distribution $X_0,\widetilde{X}_1,\ldots,\widetilde{X}_M$ is invariant under all permutations.
\end{definition}
Note that sampling $\widetilde{X}_1,\ldots,\widetilde{X}_M$ i.i.d. from $\pi$ is an exchangeable sampler. The next proposition states that all exchangeable samplers are valid.
\begin{proposition}\label{prop:ex-valid}
    If $X_0 \sim \pi$ and $\widetilde{X}_1,\ldots,\widetilde{X}_M \sim P_{X_0}$ is the output of an exchangeable sampler, then
    \begin{equation}\label{eq:valid} \mathbb{P}(p_{MC} \le \alpha) \le \alpha \text{ for all } \alpha \in [0,1]. \end{equation}
    That is, all exchangeable samplers are valid.
\end{proposition}
\begin{proof}
    Fix $\alpha \in [0,1]$ and assume that $X_0 \sim \pi$. By definition, $p_{MC}$ is less than $\alpha$ if and only if $T_0$ is in top $\alpha$ proportion of $T_0,\widetilde{T}_1,\ldots,\widetilde{T}_M$.     Since $\widetilde{X}_1,\ldots,\widetilde{X}_M$ is the output of an exchangeable sampler, $X_0,\widetilde{X}_1,\ldots,\widetilde{X}_M$ are exchangeable. This implies that $T_0,\widetilde{T}_1,\ldots,\widetilde{T}_M$ are also exchangeable. By symmetry, we can conclude that the probability that $T_0$ is in the top $\alpha$ proportion of $T_0,\widetilde{T}_1,\ldots,\widetilde{T}_M$ is at most $\alpha$.
\end{proof}

We will now study the consistency of $p_{MC}$.

\subsection{Consistency}\label{sec:consistency}

The standard Monte Carlo p-value is an estimate of what we will call the \emph{analytic p-value}. The analytic p-value is the probability of observing a value of $T(\widetilde{X})$ ``as extreme'' as $T(X_0)$ when $\widetilde{X} \sim \pi$. The analytic p-value is thus the tail probability
\begin{equation*}
    p_A = \pi\left(\{x \in \mathcal{X}: T(x) \ge T(X_0)\}\right).
\end{equation*}
Suppose we treat $X_0$ as fixed and sample $\widetilde{X}_1,\ldots,\widetilde{X}_M$ i.i.d. from $\pi$. The indicator variables $I[\widetilde{T}_i \ge T_0]$ are then i.i.d. Bernoulli random variables with expectation $p_A$. Thus, $p_{MC}$ converges in probability to $p_A$ by the weak law of large numbers. If this occurs when we use a different sampler to produce $\widetilde{X}_1,\ldots,\widetilde{X}_M$, then we will say that the sampler is consistent.
\begin{definition}
    A sampler $P_{X_0}$ is \emph{consistent} if for all $X_0 \in \mathcal{X}$, $\widetilde{X}_1,\ldots,\widetilde{X}_M \sim P_{X_0}$ implies
    \[p_{MC} \stackrel{\mathbb{P}}{\to} p_A \text{ as } M \to \infty.\]
\end{definition}
The definition of consistency is conditional on $X_0$. A sampler is consistent if and only if $p_{MC}$ converges to $p_A$ for every possible value of $X_0$. On the other hand, the definition of validity is marginal over $X_0$. The probability in \eqref{eq:valid} is defined with respect to both the randomness in $X_0 \sim \pi$ and the randomness in $\widetilde{X}_1,\ldots,\widetilde{X}_M\sim P_{X_0}$.

In Section~\ref{sec:E MCMC} we will study samplers that use Markov chains to produce $\widetilde{X}_1,\ldots,\widetilde{X}_M$. All of these samplers will be valid but, without further assumptions, they may be inconsistent.

\subsection{Composite null hypotheses}\label{sec:composite}

A limitation of Monte Carlo tests is that they cannot directly be used to test composite hypothesis \cite{dufour2006monte}. Recall that a composite hypothesis is a hypothesis of the form $X_0 \sim \pi$ for some unspecified $\pi \in \mathcal{P}$. Here $\mathcal{P}$ is a collection of distributions on $\mathcal{X}$. An example of a composite hypothesis is that $X_0 \in \mathbb{R}$ is normally distributed with mean zero and unknown variance. In contrast, a simple hypothesis is a hypothesis of the form $X_0 \sim \pi$ for a fixed distribution $\pi$. 

Monte Carlo tests do not apply directly to composite hypotheses because we cannot sample  $\widetilde{X}_i$ from the collection of distributions $\mathcal{P}$. If we sampled $\widetilde{X}_i \stackrel{\mathrm{iid}}{\sim} \pi$ but $X_0 \sim \pi'$ for distinct $\pi, \pi' \in \mathcal{P}$, then $X_0,\widetilde{X}_1,\ldots,\widetilde{X}_M$ will not be exchangeable and the Monte Carlo test may be invalid. Note that this limitation applies to both standard Monte Carlo tests and the MCMC tests that are the focus of this article. 

There are two general strategies, conditioning and maximizing, that reduce testing a composite hypothesis to testing one or more simple null hypotheses. Monte Carlo tests or other methods can then be used to test the simple null hypotheses.

Conditioning reduces the composite hypothesis to a single simple hypothesis. Conditioning applies when the collection $\mathcal{P}$ has a sufficient statistic. In this case, we can condition on the sufficient statistic $S(X_0)=s_0$ and then sample from $\pi(\cdot \mid S(X_0)=s_0)$, the distribution of $X_0 \mid S(X_0)=s_0$. By the definition of sufficiency, this conditional distribution is the same for all $\pi \in \mathcal{P}$. This approach is called \emph{co-sufficient sampling} \cite{stephens2012goodness} and is used in many of the applications in Section~\ref{sec:applications}.  

Maximizing is an alternative to co-sufficient sampling. Maximizing has the advantage of not requiring a sufficient statistic, but is computationally intensive and may suffer from low power. To describe the maximizing procedure, suppose that $p_\pi$ is a valid p-value for the simple null hypothesis $X_0 \sim \pi$. Then, the maximized p-value
\begin{equation}
    p_{\mathrm{max}} = \sup\{p_\pi : \pi \in \mathcal{P}\}, \label{eq:p-max}
\end{equation}
is a valid p-value for the composite null hypothesis $X_0 \sim \pi$ for some $\pi \in \mathcal{P}$. When $p_\pi$ is computed by sampling $\widetilde{X}_i$, the maximization procedure in \eqref{eq:p-max} may be computationally intensive or impossible. In practice, $p_{\mathrm{max}}$ is approximated by either restricting to a finite subset of $\mathcal{P}$ \cite{diciccio2022confidence} or by running a stochastic optimization algorithm such as simulated annealing \cite{dufour2006monte}.

\section{Exchangeable MCMC samplers}\label{sec:E MCMC}

Suppose now that we are unable to sample directly from $\pi$ and instead use MCMC to construct a sampler. The simplest approach would be \emph{sequential MCMC}. In sequential MCMC, we sample $\widetilde{X}_1,\ldots,\widetilde{X}_M$ with a Markov chain that starts at $X_0$ and transitions according to $K^L$ for some fixed value of $L$. If we assume the null hypothesis $X_0 \sim \pi$, then we will have $\widetilde{X}_i \sim \pi$ for every $i$ since $\pi$ is a stationary distribution of $K$. This implies that $\widetilde{X}_1,\ldots,\widetilde{X}_M$ each have the same marginal distribution as $X_0$, but it does not imply that $X_0,\widetilde{X}_1,\ldots,\widetilde{X}_M$ are exchangeable. Indeed, the joint distribution of $X_0,\widetilde{X}_1,\widetilde{X}_2$ will typically be different to the joint distribution of $X_0,\widetilde{X}_2,\widetilde{X}_1$.  

The lack of exchangeability means that we can not apply Proposition~\ref{prop:ex-valid}. Thus, sequential MCMC may be an invalid sampler. There are natural examples where $\mathbb{P}(p_{MC} \le \alpha)$ is of the order $\sqrt{\alpha}$ when $X_0 \sim \pi$ and sequential MCMC is used. Section~6 of the supplementary material of \cite{Gerry1} gives details about one such example.

The following MCMC methods are exchangeable samplers. They achieve exchangeability by using both the transition kernel $K$ and a second transition kernel $\widehat{K}$ called the reversal of $K$. 
\begin{definition}\label{def-reversal}
    Suppose that $\pi$ is a stationary distribution of $K$. A \textit{reversal of $K$} is any transition kernel $\widehat{K}$ with transition kernel density $\widehat{k}$ such that for all $x,y \in \mathcal{X}$
     \[f(x)k(x,y) = f(y)\widehat{k}(y,x), \]
     where $f$ is the density of $\pi$ and $k$ is the kernel density of $K$. If $\widehat{K}=K$, then $K$ is said to be \emph{reversible}.
\end{definition}
\begin{remark}
If $f(x) > 0$ for all $x \in \mathcal{X}$, then $K$ has a unique reversal with respect to $\pi$ with kernel density given by 
\[\widehat{k}(y,x)= \frac{f(x)}{f(y)}k(x,y). \]
We will thus often refer to $\widehat{K}$ as \emph{the} reversal of $K$. 
\end{remark}
The name reversal is justified by the following proposition (see, for example, Chapter~2.4.2 in \cite{bremaud2020markov}). 
\begin{proposition}
    Suppose that $X_0,X_1,\ldots,X_M$ is a Markov chain with transition kernel $\widehat{K}$ and initial distribution $X_0 \sim \pi$. The time-reversed sequence $X_M,X_{M-1},\ldots,X_0$ has the same distribution as a Markov chain with transition kernel $K$ and initial distribution $X_M \sim \pi$.
\end{proposition}
Equipped with Definition~\ref{def-reversal}, we are now ready to describe and study MCMC exchangeable samplers.

\subsection{The parallel method}\label{sec:Parallel}

Besag and Clifford introduced two methods for MCMC significance tests in \cite{BC}. The first was the \emph{parallel method}, which is an exchangeable sampler that proceeds as follows.

\begin{enumerate}
    \item Starting from $X_0$, run a Markov chain according to $\widehat{K}$ for $L$ steps to sample $X^*$.
    \item Independently, for $i =1,\ldots,M$,  starting from $X^*$, run a Markov chain  according to $K$ for $L$ steps to sample $\widetilde{X}_i$.
\end{enumerate}

The parallel method is also called the \emph{hub-and-spoke sampler} \cite{GoF} due to the graphical representation of the method shown in Figure~\ref{fig parallel}.

\begin{figure}[hb]
   \centering
    \begin{tikzpicture}[line cap=round,line join=round,>=triangle 45,x=0.6cm,y=0.6cm]
        \draw (6,1) node {$X^*$};
        \draw [-Latex,line width=1.pt,dash pattern=on 3pt off 3pt] (3.,1.) -- (5.579905262245591,0.9713571769712903);
        \draw [-Latex,line width=1.pt] (6.017172234999721,1.420719757493173) -- (6.,4.);
        \draw [-Latex,line width=1.pt] (6.42,0.97) -- (9.,1.);
        \draw [-Latex,line width=1.pt] (6.,0.578929934571454) -- (6.,-2.);
        \draw [-Latex,line width=1.pt] (5.7243246731124655,0.6817184985811287) -- (3.8721854641797337,-1.1148062088881776);
        \draw [-Latex,line width=1.pt] (6.302502476389131,0.7070968559771968) -- (8.112116093706447,-1.1304848290251255);
        \draw [-Latex,line width=1.pt] (6.314795410698285,1.2796495117165376) -- (8.121320343559642,3.121320343559643);
        \draw [-Latex,line width=1.pt] (5.7068924752102665,1.3023044473897722) -- (3.8776336184515734,3.1202737895076282);
        \draw (3,1) node [anchor = east] {$X_0$};
        \draw (4,-1) node[anchor = north east] {$\widetilde{X}_1$};
        \draw (6,-2) node[anchor=north] {$\widetilde{X}_2$};
        \draw (8,-1) node[anchor=north west] {$\widetilde{X}_3$};
        \draw (9,1) node[anchor=west] {$\ldots$};
        \draw (8,3) node[anchor=south west] {$\widetilde{X}_{M-2}$};
        \draw (6,4) node[anchor=south] {$\widetilde{X}_{M-1}$};
        \draw (4,3) node[anchor=south east] {$\widetilde{X}_M$};
        \end{tikzpicture}
            
    \caption{A graphical representation of the parallel method. The solid arrows represent running the Markov chain forwards for $L$ steps and the dashed arrow represents running the Markov chain backwards for $L$ steps.}
    
    \label{fig parallel}
\end{figure}
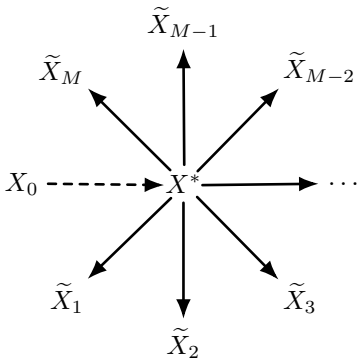

\begin{proposition}\label{prop:parallel}
    The parallel method is an exchangeable sampler.
\end{proposition}

\begin{proof}
    First note that $X_0,\widetilde{X}_1,\ldots,\widetilde{X}_M$ are conditionally independent given $X^*$. Also, for $1 \le i \le M$
    \[\widetilde{X}_i \mid X^* \sim K^L(X^*,\cdot). \]
    Now assume that $X_0 \sim \pi$. Since $\widehat{K}$ is the reversal of $K$, we know that  $(X^*, X_0) \stackrel{d}{=} (Y_0,Y_1)$ where $Y_0 \sim \pi$ and $Y_1 \mid Y_0 \sim K^L(Y_0,\cdot)$. Thus.
     \[X_0 \mid X^* \sim K^L(X^*,\cdot).\] 
     It follows that $X_0,\widetilde{X}_1,\ldots, \widetilde{X}_M$ are independent and identically distributed given $X^*$. By marginalizing over $X^*$, we get that $X_0,\widetilde{X}_1,\ldots, \widetilde{X}_M$ are exchangeable.
\end{proof}
 By Proposition \ref{prop:parallel}, the parallel method can be used to conduct valid MCMC tests. However, for small values of $L$, each sample $\widetilde{X}_i$ will be strongly dependent on $X_0$. This dependency means that the parallel method may be inconsistent.
\begin{proposition}\label{prop:inconsitent}
    Suppose $\widetilde{X}_1,\ldots,\widetilde{X}_M$ are generated via the parallel method. If the number of steps, $L$, is fixed and the number of samples, $M$, goes to infinity, then there exists a random variable $p_\infty$ such that 
    \[p_{MC} \stackrel{\mathbb{P}}{\to} p_{\infty}, \]
    and
    \begin{equation}
        p_{\infty} \stackrel{d}{=} K^L\left(X^*, \{x : T(x) \ge T(X_0)\}\right),\label{eq:p-infty}
    \end{equation}
    where $X^* \sim \widehat{K}^L(X_0,\cdot)$. 
\end{proposition}
\begin{proof}
    Recall that in step 1 of the parallel method, we draw $X^*\sim\widehat{K}^L(X_0,\cdot)$. Next, we draw i.i.d. samples $\widetilde{X}_1,\ldots,\widetilde{X}_M \stackrel{\text{iid}}{\sim} K^L(X^*,\cdot)$. If $X^*$ is fixed and $M$ goes to infinity, $p_{MC}$ will converge to the analytic p-value for the distribution $K^L(X^*,\cdot)$. By averaging over the possible values of $X^* \sim \widehat{K}^L(X_0,\cdot)$, we get the expression for $p_\infty$ in \eqref{eq:p-infty}.
\end{proof}
In words, Proposition~\ref{prop:inconsitent} says that $p_{MC}$ converges to a mixture of analytic p-values. The limiting value is thus random since it depends on the sample $X^*$. Some consequences of this are discussed in Section~\ref{sec:examples} where we use the parallel method to test the null hypothesis that $X_0$ is normally distributed with mean zero and variance one.

\subsection{The permuted serial method}\label{serial}

Besag and Clifford \cite{BC} also introduced the serial method. Here we will present a variation called the \emph{permuted serial method} due to \cite{GoF}.  This method requires the same amount of computation as the parallel method but uses longer runs of the Markov chain to sample $\widetilde{X}_1,\ldots,\widetilde{X}_M$.

\begin{enumerate}
    \item Sample a permutation $\sigma$ uniformly from the set of all permutations of $\{0,1,\ldots,M\}$.\label{serial:1}
    \item Set $m^* = \sigma(0)$ and $Y_{m^*} = X_0$.\label{serial:2}
    \item \begin{enumerate}
        \item From $Y_{m^*}$ run a Markov chain according to $\widehat{K}^L$ for $m^*$ steps to sample $Y_{m^*-1}$, $Y_{m^*-2}$, $\ldots$, $Y_0$.
        \item From $Y_{m^*}$ run a Markov chain according to $K^L$ for $M-m^*$ steps to sample $Y_{m^*+1}$, $Y_{m^*+2}$, $\ldots$, $Y_M$.
    \end{enumerate}\label{serial:3}
    \item For $i=1,\ldots,M$, set $\widetilde{X}_i =Y_{\sigma(i)}$. \label{serial:4}
\end{enumerate}

The permuted serial method is represented visually in Figure~\ref{fig serial}.

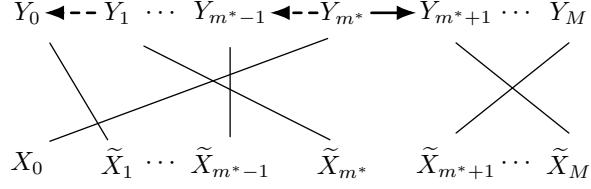
\begin{figure}[t]
    \centering
    \begin{tikzpicture}[line cap=round,line join=round,>=triangle 45,x=1.5cm,y=2.0cm]
        \draw (0,0) node {$Y_{m^*}$};
        \draw (1,0) node {$Y_{m^*+1}$};
        \draw (1.55,0) node {$\cdots$};
        \draw (2,0) node {$Y_{M}$};
        \draw (-1.6,0) node {$\cdots$};
        \draw (-1,0) node {$Y_{m^*-1}$};
        \draw (-2,0) node {$Y_1$};
        \draw (-2.8,0) node {$Y_0$};
        
        \draw (0,-1) node {$\widetilde{X}_{m^*}$};
        \draw (1,-1) node {$\widetilde{X}_{m^*+1}$};
        \draw (1.55,-1) node {$\cdots$};
        \draw (2,-1) node {$\widetilde{X}_{M}$};
        \draw (-1.6,-1) node {$\cdots$};
        \draw (-1,-1) node {$\widetilde{X}_{m^*-1}$};
        \draw (-2,-1) node {$\widetilde{X}_1$};
        \draw (-2.8,-1) node {$X_0$};
        \draw [-Latex,line width=1.pt] (0.25,0) -- (0.65,0);
        \draw [-Latex,line width=1.pt,dash pattern=on 3pt off 3pt] (-0.23,0) -- (-0.65,0);
        \draw [-Latex,line width=1.pt,dash pattern=on 3pt off 3pt] (-2.2,0.) -- (-2.65,0.);
        \draw [line width=0.5pt] (-0.13790404402645942,-0.17091664160820047) -- (-2.6,-0.8544421677115974);
        \draw [line width=0.5pt] (2,-0.2) -- (1,-0.8);
        \draw [line width=0.5pt] (1.,-0.2) -- (2,-0.8);
        \draw [line width=0.5pt] (-2.6,-0.2010721795245268) -- (-2.08796216261557,-0.8443903217394887);
        \draw [line width=0.5pt] (-1.7663030915080877,-0.22117587146874437) -- (-0.11780035208224171,-0.8443903217394887);
        \draw [line width=0.5pt] (-1.0023627976278178,-0.23122771744085316) -- (-1.0023627976278178,-0.8242866297952711);
        \end{tikzpicture}
    \caption{A graphical representation of the permuted serial method. Solid arrows represent running a Markov chain forwards for $L$ steps and the dashed arrows represent running the Markov chain backwards for $L$ steps. The lines represent the random permutation $\sigma$.}
    
    \label{fig serial}
\end{figure}
\begin{remark}
    The original version of the serial method presented in \cite{BC} differs from the above procedure. Instead of sampling a permutation $\sigma$, \cite{BC} sample $m^* \sim \mathrm{Unif}(\{0,\ldots,M\})$ and then continue to steps~\ref{serial:2} and \ref{serial:3}. They skipped step \ref{serial:4} and returned the p-value, 
\[p=\frac{\left|\{0\le i \le M : i \neq m^*, T(Y_i) \ge T(X_0)  \}\right|+1}{M+1}. \]
The above quantity is invariant under permuting the samples $\{Y_i : i \neq m^*\}$. Thus, the original serial method produces the same p-value as the permuted version and the two are equivalent for testing purposes. However, only the permuted version produces an exchangeable sequence \cite{GoF}.
\end{remark}
\begin{proposition}\label{prop:permuted serial}
    The permuted serial method is an exchangeable sampler.
\end{proposition}
To prove that the permuted serial method is an exchangeable sampler, we will use the following lemma.
\begin{lemma}\label{lemma:indep}
    Let $\sigma$ be the random permutation and $(Y_i)_{i=0}^M$ be the random samples from steps \ref{serial:1}-\ref{serial:3} of the permuted serial method. If $X_0 \sim \pi$, then $(Y_i)_{i=0}^M$ are independent of $\sigma$.
\end{lemma}
\begin{proof}
    Assume that $X_0 \sim \pi$. The random permutation $\sigma$ only affects $(Y_i)_{i=0}^M$ through the random index $m^*=\sigma(0)$ at which we set $Y_{m^*} = X_0$. Thus, to show that $(Y_i)_{i=0}^M$ is independent of $\sigma$, we will show that the distribution of $(Y_i)_{i=0}^M$ does not depend on $m^*$. 
    
    To do this, recall that $Y_{m^*-1},Y_{m^*-2},\ldots,Y_0$ are generated by running a Markov chain from $Y_{m^*}=X_0$ according to $\widehat{K}^L$. By the assumption that $X_0 \sim \pi$ and the definition of $\widehat{K}$, this is equivalent to sampling $Y_0 \sim \pi$ and then running a chain from $Y_0$ according to $K^L$ to sample $Y_1,Y_2,\ldots,Y_{m^*}$. The samples $Y_{m^*+1},Y_{m^*+2},\ldots,Y_M$ are generated by running a Markov chain from $Y_{m^*}=X_0$ according to $K^L$. Thus, the distribution of $(Y_i)_{i=0}^M$ is equivalent to sampling $Y_0 \sim \pi$ and running a Markov chain from $Y_0$ according to $K^L$ to sample $Y_1,\ldots,Y_M$. This description of the distribution of $(Y_i)_{i=0}^M$ does not depend on $m^*$. Thus, $(Y_i)_{i=0}^M$ are independent of $\sigma$. 
\end{proof}
We will now use Lemma~\ref{lemma:indep} to prove that the permuted serial method produces exchangeable samples.
\begin{proof}[Proof of Proposition~\ref{prop:permuted serial}]
   Assume that $X_0 \sim \pi$ and let $\sigma$, $(Y_i)_{i=0}^M$ and $(\widetilde{X}_i)_{i=1}^M$ be as in the permuted serial method. By Lemma~\ref{lemma:indep}, $(Y_i)_{i=0}^M$ and $\sigma$ are independent. Since $X_0 = Y_{\sigma(0)}$ and $\widetilde{X}_i = Y_{\sigma(i)}$, we can write the distribution of $X_0,\widetilde{X}_1,\ldots,\widetilde{X}_M$ in terms of the distribution of $\sigma$ and $(Y_i)_{i=0}^M$. Specifically, if $P_X$ is the distribution of $(X_0,\widetilde{X}_1,\ldots,\widetilde{X}_M)$ and $P_Y$ is the distribution of $(Y_0,\ldots,Y_M)$, then
   \begin{align*}
    &P_X(X_0,\widetilde{X}_1,\ldots,\widetilde{X}_M) \\ 
    &= \sum_{\sigma \in S_{0:M}} \frac{P_Y\left(X_{\sigma^{-1}(0)}, \widetilde{X}_{\sigma^{-1}(1)},\ldots,\widetilde{X}_{\sigma^{-1}(M)}\right)}{(M+1)!} ,
   \end{align*}
   where  $S_{0:M}$ is the set of permutations of $\{0,1,\ldots,M\}$. The right-hand-side of the above equation is invariant under permuting $X_0,\widetilde{X}_1,\ldots,\widetilde{X}_M$. Thus, $X_0,\widetilde{X}_1,\ldots,\widetilde{X}_M$ are exchangeable.
\end{proof}

Under an additional assumption on the Markov chain $K$, the permuted serial method is consistent if we keep $L$ fixed and send $M$ to infinity. Intuitively, this is because the number of steps of the Markov chain between samples $\widetilde{X}_i$ tends to increases with $M$. The dependencies among $\widetilde{X}_1,\ldots,\widetilde{X}_M$ are thus decreasing with $M$. This in turn implies that a weak law of large numbers may be applied to the random variables $I[T(\widetilde{X}_i) \ge T(X_0)]$. More formally,
\begin{proposition}\label{prop:serial consistent}
    Assume that $K^L$ is irreducible meaning that for all $x,y \in \mathcal{X}$, there exists $j$ such that $k^{Lj}(x,y) > 0$. Suppose that $\widetilde{X}_1,\ldots,\widetilde{X}_M$ are sampled using the permuted serial method with $L$ fixed. Then, as $M \to \infty$,
    \[p_{MC} \stackrel{\mathbb{P}}{\to} p_A. \]
    Thus, the permuted serial method is consistent. 
\end{proposition}
The proof of Proposition~\ref{prop:serial consistent} is an application of the ergodic theorem for Markov chains (Theorem~\ref{theorem:ergodic}). The details are given in the \hyperref[appn]{Appendix}.

\subsection{Tree methods}\label{sec:tree-MCMC}

Tree methods are a family of exchangeable MCMC samplers that unify the parallel and permuted serial methods. The material in this section is new, although similar ideas appear in Section~7 of \cite{Gerry2}. 

Tree methods can be motived by looking again at the visual representations of the parallel and permuted serial methods in Figures~\ref{fig parallel} and \ref{fig serial}. Both methods take as input the data $X_0$ and then generate a dependent collection of samples $\widetilde{X}_1,\ldots,\widetilde{X}_M$. This collection of samples is created by picking a random location for $X_0$ and then using the Markov chains $K$ and $\widehat{K}$ to  generate $\widetilde{X}_1,\ldots,\widetilde{X}_M$. The exact procedure used to generate $\widetilde{X}_1,\ldots,\widetilde{X}_M$ is determined by the placement and direction of the arrows in Figures~\ref{fig parallel} and \ref{fig serial}.

The arrows connecting samples in Figures~\ref{fig parallel} and \ref{fig serial} can be represented as a \emph{directed tree}. Recall that a directed graph is a pair $(V,E)$ where $V$ is a finite set of vertices and $E$ is set of directed edges $E\subseteq V \times V$. For $v,u \in V$, $(v,u) \in E$ means that there is a directed edge from $v$ to $u$ in the graph $(V,E)$. A directed tree $\mathcal{T}$ is a directed graph such that when we ignore the direction of the edges $\mathcal{T}$ is connected and does not contain any cycles, double edges or self-loops. 

The input to a tree method is a \emph{marked tree}. A marked tree is a directed tree with $M+1$ distinguished vertices. An example is shown in Figure~\ref{fig trees}. 

\begin{definition}
    A \emph{marked tree} is a pair $(\mathcal{T},\ell)$ where $\mathcal{T}=(V,E)$ is a directed tree and $\ell$ is an injective function from $\{0,1,\ldots,M\}$ to $V$. The vertices $\{\ell(0),\ell(1),\ldots,\ell(M)\}$ are called \emph{marked vertices}.
\end{definition}
In a marked tree $(\mathcal{T},\ell)$, the vertices of $\mathcal{T}$ correspond to samples and the edges correspond to steps of the Markov chain. The function $\ell$ allows us to use auxiliary samples in the tree method. The marked vertices  correspond to the samples $X_0, \widetilde{X}_1,\ldots,\widetilde{X}_M$ used in the Monte Carlo test whereas unmarked vertices correspond to auxiliary samples. For example, the variable $X^*$ at the ``hub'' of the parallel method corresponds to an unmarked vertex as it is used to generate $\widetilde{X}_1,\ldots,\widetilde{X}_M$, but we do not compare $X_0$ to $X^*$. Unmarked vertices can also be used to simulate running the Markov chain for $L$ steps. This is done by placing $L-1$ unmarked vertices between two marked vertices.

Given a marked tree $(\mathcal{T},\ell)$, the tree method generates an exchangeable sample by picking a random marked vertex and then ``exploring'' the tree $\mathcal{T}$ by taking steps of the Markov chains $K$ and $\widehat{K}$. This generates a sample $Y_v$ for each vertex $v \in V$. The samples from marked vertices $(Y_{\ell(i)})_{i=0}^M$ are then permuted to create the final sample $X_0,\widetilde{X}_1,\ldots,\widetilde{X}_M$. More formally,
\begin{enumerate}
    \item Sample $\sigma$, a uniform permutation of  $\{0,1,\ldots,M\}$.
    \item Set $U=\{\ell(m^*)\}$ and $Y_{\ell(m^*)}=X_0$.
    \item While $U \neq V$:
        \begin{enumerate}
            \item Find $v \in V \setminus U$ and $u \in U$ with $(u,v) \in E$ or $(v,u) \in E$.
            \item If $(u,v) \in E$, sample $Y_v \sim K(Y_u,\cdot)$.
            \item If $(v,u) \in E$, sample $Y_v \sim \widehat{K}(Y_u,\cdot)$.
            \item Set $U=U\cup\{v\}$.
        \end{enumerate}
    \item For $i=1,\ldots,M$, set $\widetilde{X}_i =Y_{\ell(\sigma(i))}$.
\end{enumerate}

\begin{figure}[t]
    \centering
    \begin{tikzpicture}[line cap=round,line join=round,>=triangle 45,x=1.5cm,y=2cm]

        \draw [->,line width=1.pt] (2.2,2.) -- (2.9,2.);
        \draw [->,line width=1.pt] (3.1,2.1) -- (3.38,2.38);
        \draw [->,line width=1.pt] (3.15,2.) -- (3.8,2.);
        \draw [->,line width=1.pt] (5.1,2.1) -- (5.38,2.38);
        \draw [->,line width=1.pt] (4.2,2) -- (4.9,2);
        \draw [->,line width=1.pt] (5.15,2) -- (5.8,2);
        
        \draw  (2,2.) circle (6pt);
        \draw  (3.5,2.5) circle (6pt);
        \draw  (3.,2.) [fill=black] circle (2.5pt);
        \draw  (4.,2.) circle (6pt) ;
        \draw  (5.5,2.5) circle (6pt);
        \draw  (5.,2) [fill=black] circle (2.5pt);
        \draw  (6.,2) circle (6pt);

        \draw  (2.,2.)  node {0};
        \draw  (3.5,2.5) node {1};
        \draw  (4.,2.)  node {2};
        \draw  (5.5,2.5) node {3};
        \draw  (6.,2) node {4};
        
        \end{tikzpicture}
    \caption{A marked tree with $M=4$. There are $M+1=5$ marked vertices labeled $0,1,\ldots,M$. There are two unmarked vertices represented as black dots. Travelling in the direction of an edge corresponds to one step of the Markov chain $K$ and travelling against the direction of an edge corresponds to a step of the reversal $\widehat{K}$.}
    \label{fig trees}
\end{figure}
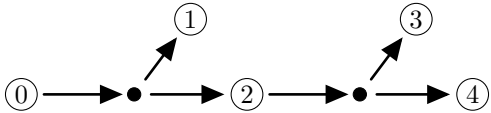
It is important that the graph $\mathcal{T}=(V,E)$ is a tree. Since $\mathcal{T}$ is a tree, there is a unique undirected path from the vertex $\ell(m^*)$ to any other vertex $v$. This means that once we have picked the starting vertex $\ell(m^*)$, the distribution of $(Y_v)_{v \in V}$ is determined. The order in which we add vertices to the set $U$ does not affect the distribution of $(Y_v)_{v \in V}$. The existence of a unique path from $\ell(m^*)$ to $v$ also implies that the above procedure will always terminate. 

We will now show that tree methods produce exchangeable samples which can be used to in a Monte Carlo test.
\begin{proposition}\label{prop:tree}
    Let $(\mathcal{T},\ell)$ be a marked tree, then the tree method for $(\mathcal{T},\ell)$ is an exchangeable MCMC sampler.
\end{proposition}
\begin{proof}
    Let $\sigma$ and $(Y_v)_{v \in V}$  be the permutation and samples produced by running the tree method for $(\mathcal{T},\ell)$. The samples $X_0,\widetilde{X}_1,\ldots,\widetilde{X}_M$ are a uniform permutation of $(Y_{\ell(i)})_{i=0}^M$. Thus, as in the proof of Proposition~\ref{prop:permuted serial}, it suffices to show that if $X_0 \sim \pi$, then $(Y_v)_{v \in V}$ are independent of the permutation $\sigma$. 

    If $X_0 \sim \pi$, then $Y_v \sim \pi$ for all $v \in V$ by stationarity. Furthermore, by reversibility, we know that for every edge $(u,v) \in E$, the distribution of $(Y_u,Y_v)$ does not depend on whether we first sample $Y_u$ or $Y_v$. By induction, this implies that the process used to generate $(Y_v)_{v \in V}$ is equivalent to deterministically starting at $Y_{\ell(0)} =X_0$ and then exploring $\mathcal{T}$ from the vertex $\ell(0)$. Thus, when $X_0 \sim \pi$, the distribution of $(Y_v)_{v \in V}$ does not depend on $\sigma$ and so $X_0,\widetilde{X}_1,\ldots,\widetilde{X}_M$ are exchangeable. 
\end{proof}

Although tree methods have a conceptual benefit of unifying existing MCMC samplers, they may not always be computationally efficient. This is because the process of finding edges connecting $U$ to $V\setminus U$ may be costly. The naive approach of checking every possible edge has computational cost that is quadratic in the number of vertices. However, certain tree method are tractable. These include the parallel and permuted serial methods, as well as versions of these methods where the step size $L$ is allowed to vary between samples. The split-star method in Section 7 of \cite{Gerry2} is another example of a tree method that can be performed efficiently.

\subsection{Comparison}

In this section we offer some heuristic arguments comparing the parallel and permuted serial methods from \cite{BC}. The main benefit of the permuted serial method over the parallel method is its consistency. Even for finite $M$, we expect the permuted serial method to be closer to a standard Monte Carlo test than the parallel method. In the parallel method, each $\widetilde{X}_i$ is the result of running a Markov chain $2L$ steps from $X_0$. However, in the permuted serial method, the number of steps is $|\sigma(i)-\sigma(0)|L$ where $\sigma$ is a uniformly drawn random permutation of $\{0,1,\ldots,M\}$. Thus, $\widetilde{X}_i$ is an average of $\frac{M+2}{3}L$ steps away from $X_0$. In applications, the number of MCMC samples $M$ will be in the hundreds or higher. Thus, the dependency between $X_0$ and $\widetilde{X}_i$ in the permuted serial method can be substantially less than in the parallel method. This decrease in dependency means that the permuted serial tests should behave more like a standard Monte Carlo test. 

On the other hand, the parallel method has a computational advantage. Both methods require $M+1$ evaluations of the test statistic $T$ and roughly $ML$ steps of the Markov chain. However, as suggested by its name, the parallel method can be parallelized. Once the latent value $X^*$ has been sampled, $\widetilde{X}_1,\ldots,\widetilde{X}_M$ can be computed in parallel. When parallel computation is available, we can run the parallel method with a much larger value of $L$. Increasing $L$ decreases the dependency between samples. Thus, by parallelizing the computation of $\widetilde{X}_1,\ldots,\widetilde{X}_M$ and using a larger value of $L$, we can get an improvement over using the permuted serial method.

The split-star method from \cite[Section~7]{Gerry2} combines the benefits of both the permuted serial and parallel methods. The split-star method is a tree method that essentially performs the permuted serial method $M_1$ times to produce a total of $M_1M_2$ samples $\widetilde{X}_i$. The authors of \cite{Gerry2} recommend setting $M_1$ equal to the number of parallel processors available so that each run of the serial method can be performed concurrently. The method thus benefits from long runs of the Markov chain like the permuted serial method, but it can be parallelized \cite{Gerry2}.

\section{Examples}\label{sec:examples}

In this section, we study two testings problems that highlight the difference between MCMC p-values and standard Monte Carlo p-values. In first example we examine the power of the parallel method when testing the mean of a normally distributed random variable with known variance. As in Section~6 of \cite{BC} we use an autoregressive Markov chain and examine how the choice of Markov chain affects power. In the second example, the null distribution is a bimodal mixture and the Markov chain is slow to mix. This example shows that there can be qualitative differences between MCMC tests and standard Monte Carlo tests.

\subsection{Testing the standard normal distribution}\label{sec:AR}

Suppose our null hypothesis is $X_0 \sim \mathcal{N}(0,1)$ where $\mathcal{N}(\mu,\tau^2)$ represents the normal distribution with mean $\mu$ and variance $\tau^2$. There is no need to use MCMC to test this hypothesis, but it is illustrative to study these methods in this familiar setting. In Section~6 of \cite{BC}, the authors study the power of the parallel test when using an autoregressive Markov chain to generate the MCMC samples $\widetilde{X}_1,\ldots,\widetilde{X}_M$. In this section we will derive a complementary result. As in Section~6 of \cite{BC}, we will use the following Markov chain. 
\begin{definition}\label{def:AR}
    An order one autoregressive process with autocorrelation $\rho \in (-1,1)$, is a Markov chain $(X_i)_{i=0}^M$ that transitions from $X_i$ to $X_{i+1}$ according to,
    \[X_{i+1} = \rho X_i +\sqrt{1-\rho^2}Z_i, \]
    where  $Z_i \stackrel{\mathrm{iid}}{\sim} \mathcal{N}(0,1)$. Equivalently, the Markov chain has a transition kernel $K$ with density
    \[k(x,y) = \varphi\left(y;\rho x, 1-\rho^2\right), \] 
    where $\varphi(\cdot;\mu,\tau^2)$  is the density of the normal distribution with mean $\mu$ and variance $\tau^2$.
\end{definition}
Autoregressive Markov chains are well-studied in time series models. For any $\rho \in (-1,1)$ the Markov chain in Definition~\ref{def:AR} is reversible, and the standard normal distribution is the unique stationary distribution. However, for $\rho$ close to 1, the correlation between samples is high, and the chain is slow to mix. 

Suppose we use an order one autoregressive Markov chain in the parallel method to test $X_0 \sim \mathcal{N}(0,1)$. We will use the test statistic $T(x)=x$, so that large value of $X_0$ correspond to evidence against the null hypothesis. The analytic p-value is simply the tail probability $p_A = 1-\Phi(X_0)$, where $\Phi$ is the cumulative distribution function for $\mathcal{N}(0,1)$. However, if we use the parallel method, then Proposition~\ref{prop:inconsitent} implies the following.
\begin{proposition}
    \label{prop:non-consistent} Suppose that $\widetilde{X}_1,\ldots,\widetilde{X}_M$ are sampled by using the parallel method with an order one autoregressive Markov chain with correlation $\rho$. If the number of steps is fixed at $L$ and the number of samples $M$ goes to infinity, then 
    \[p_{MC} \stackrel{\mathbb{P}}{\to} p_\infty \stackrel{d}{=} 1 - \Phi\left(\sqrt{1-\rho^{2L}}X_0 - \rho^LZ^*\right), \]
    where $Z^* \sim \mathcal{N}(0,1)$ and $Z^*$ is independent of $X_0$.
\end{proposition}
Note that the limiting distribution depends only on $\rho^L$ and not on $\rho$ or $L$ individually. This is expected since running this autoregressive chain with autocorrelation $\rho$ for $L$ steps is equivalent to running a chain with autocorrelation $\rho^L$ for one step. 

From Proposition~\ref{prop:non-consistent}, we can study the limiting power of the parallel method. Specifically, consider the alternative $X_0 \sim \mathcal{N}(\mu,1)$ with $\mu > 0$. The uniformly most powerful level $\alpha$ test for this alternative rejects the null when $X_0 \ge \Phi^{-1}(1-\alpha)$ or equivalently when $p_A \le \alpha$. The power of this analytic test is $1-\Phi(\Phi^{-1}(1-\alpha) - \mu)$. On the other hand, if we reject the null when $p_\infty \le \alpha$, then our power is 
\begin{equation}\label{eq:power}1-\Phi(\Phi^{-1}(1-\alpha) - \sqrt{1 - \rho^{2L}}\mu). \end{equation}
Thus, as shown in \cite{BC}, the dependencies in a MCMC test reduce the signal strength, $\mu$, by a factor of $\sqrt{1-\rho^{2L}}$. For $\rho^L$ close to $1$, this loss in power can be substantial as shown in Figure~\ref{fig:power}.

\begin{figure}[ht]
    \centering
    \includegraphics*{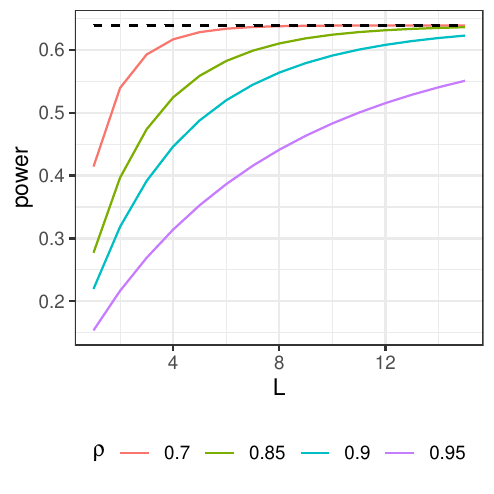}
    \caption{The limiting power of the parallel method from \eqref{eq:power} is shown as a function of $L$ for $\mu=2$, $\alpha = 0.05$ and various values of $\rho$. The dashed line in black is $1-\Phi(\Phi^{-1}(0.95)-2)$, the highest possible power at level $\alpha$. When $\rho=0.7$, the power of the Monte Carlo test quickly approaches the optimal power. When $\rho$ is close to $1$, this convergence is slow.
    }
    \label{fig:power}
\end{figure}

This example also shows the importance of choosing the Markov chain $K$. Here $K$ is determined by the choice of the autocorrelation $\rho$, but in other situations there could be a variety of Markov chains available. The results in this section suggest that Markov chains with faster mixing times will produce more powerful tests. Thus, one can use standard diagnostics from applied Bayesian statistics such as trace plots and effective sample sizes to assess the mixing time. These methods can also be used to pick the parameter $L$. However, for the MCMC samples to remain valid all diagnostics must be independent of the data $X_0$. 

\subsection{Testing a bimodal distribution}

Let $\mathcal{X}$ be the discrete set $\{1,2,\ldots,100\}$ and define $\pi$ by
\[\pi(x) \propto \frac{1}{2}\varphi\left(x;25,6^2\right) + \frac{1}{2}\varphi\left(x;75,6^2\right), \]
where again $\varphi(\cdot;\mu,\tau^2)$ is the density of the normal distribution with mean $\mu$ and variance $\tau^2$. The distribution $\pi$ has two modes, one at $x=25$ and another at $x=75$. As before, let $T(x)=x$ so that large values of $x$ correspond to evidence against the hypothesis $X_0 \sim \pi$. Again the analytic p-value is the tail probability
\[p_A = \sum_{x \ge X_0}\pi(x). \] 
The analytic p-value is a decreasing function of $X_0$ and $p_A \le 0.05$ if and only if $X_0 \ge 83$. The density $\pi(x)$ and the rejection region $\{X_0 \ge 83\}$ are shown in Figure~\ref{fig:bimodal}. 

\begin{figure}[ht]
    \centering
    \includegraphics*{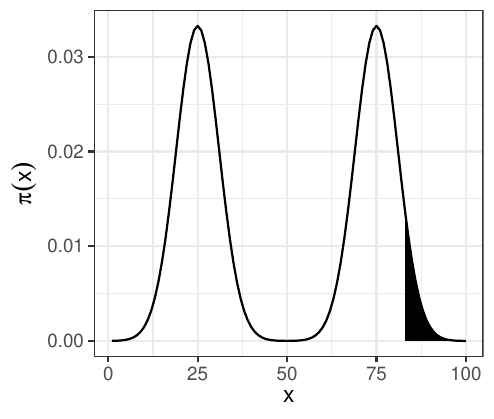}
    \caption{The density $\pi$ has modes at $x=25$ and $x=75$. The shaded region shows the 95th percentile of $\pi$.}
    \label{fig:bimodal}
\end{figure}

We will use Metropolis--Hastings to construct a reversible Markov chain with stationary distribution $\pi$. Our proposal distribution will be $x'=x \pm 1$ with equal probability. Applying Metropolis--Hastings to this proposal produces a reversible transition kernel $K$ with stationary distribution $\pi$. If $a(x,x')=\min\{1, \pi(x')/\pi(x)\}$, then the transition kernel $K$ is given by
\[x \mapsto y = \begin{cases}
    x+1 &\text{w.p. } \frac{a(x,x+1)}{2},\\
    x-1 &\text{w.p. } \frac{a(x,x-1)}{2},\\
    x &\text{w.p. } 1 - \frac{a(x,x+1)+a(x,x-1)}{2}.
\end{cases}\]
Since $\pi$ is bimodal and our proposal distribution takes small steps, $K$ will be very slow to mix. Specifically, a Markov chain initialized near one mode will take a very long time to transition to a point near the other mode. Thus, when using $K$ to perform an MCMC test, we would expect $\widetilde{X}_1,\ldots,\widetilde{X}_M$ to be highly dependent on which mode $X_0$ is closer to. For example, if $X_0 = 40$, then it is very unlikely that any of $\widetilde{X}_1,\ldots,\widetilde{X}_M$ will lie in $\{50,\ldots,100\}$. The MCMC tests are thus likely to reject the null hypothesis when $X_0 = 40$ whereas a standard Monte Carlo test essentially never rejects when $X_0 = 40$. 

This intuition is confirmed by a simulation. The results of which are summarized in Table~\ref{table:bimodal}. In the simulation we drew $X_0 \sim \pi$ 2500 times and computed $p_{MC}$ using  different samplers. The bottom row reports the percentage of simulations when $p_{MC} \le 0.05$. This is the proportion of times we would reject $X_0 \sim \pi$. By the validity of these sampler, each entry in bottom row should be at most $5\%$. The first and second rows show the percentage of simulations when $p_{MC} \le 0.05$ and $X_0 \le 50$ and $X_0 > 50$ respectively. The standard Monte Carlo test does not reject when $X_0 \le 50$ since the rejection region for the analytic p-value is $\{X_0 \ge 83\}$. On the other hand, the MCMC tests reject when $X_0$ is in the right tail of either mode. For both MCMC methods, the percentage of rejections when $X_0 \le 50$ is close to the percentage of rejections when $X_0 > 50$. By marginal validity, these two percentages should be close to $2.5\%$.

\begin{table}[ht]
    \centering 
    \caption{MCMC p-values and standard Monte Carlo p-values were calculated for $2500$ values of $X_0 \sim \pi$. The table reports the percentage of simulations for which $p_{MC} \le 0.05$. For all methods $M=99$ samples were taken. For the MCMC methods, $L=100$ steps were used. Code to reproduce this table and the other examples can be found at \url{https://github.com/Michael-Howes/MCMC-significance-tests}.}
    \label{table:bimodal}
    \begin{tabular}{c|ccc}
         & Standard & Parallel & Permuted Serial \\
        \hline 
      $X_0 \le 50$ & 0\% & 2.4\% & 2.6\% \\
    $X_0 >50$ & 4.4\% & 2.2\% & 2.0\% \\
        \hline 
        & 4.4\% & 4.6\% & 4.6\% 
        \end{tabular}  
\end{table}

In general, if we use an MCMC test with a slow to mix Markov chain, then we should expect the MCMC test to behave very differently to a standard Monte Carlo test. When we compare $X_0$ to $\widetilde{X}_1,\ldots,\widetilde{X}_M$ in a standard Monte Carlo test, we are measuring how ``globally extreme'' $X_0$ is. This is because $T(X_0)$ is compared to $T(\widetilde{X})$ for $\widetilde{X}$ a typical sample from $\pi$. On the other hand, in an MCMC test, we are measuring how ``locally extreme'' $X_0$ is. We are comparing $T(X_0)$ to $T(\widetilde{X})$ where $\widetilde{X}$ depends on $X_0$. When the Markov chain $K$ is slow to mix, these two forms of ``extreme'' may differ substantially as the current example shows.

The fact that MCMC tests reject when $X_0$ is locally extreme can be an advantage. In this example, we see that the MCMC tests have power against alternatives that put high probability on the interval $\{40,\ldots,50\}$ whereas the standard test has no power against these alternatives. Further consequences of the difference between locally and globally extreme values of $X_0$ are explored in \cite{Gerry1,Gerry-reply,Gerry-critique} where MCMC methods are used to detect gerrymandering as discussed later in Section~\ref{politics}.

\section{Applications}\label{sec:applications}

Since their appearance in \cite{BC}, MCMC significance tests have been applied in many domains  \cite{GoF, CPT, Gerry1, Gerry2,  yash,ramdas2023permutation}. Here we will discuss some of these applications.

\subsection{The Rasch model}\label{rasch}

The Rasch model is a model for binary matrices used for item-response data. Suppose $I$ subjects answer a survey with $J$ yes/no questions. The results from the survey can be encoded in an $I \times J$ binary matrix $A$ with $A_{ij} = 1$ if and only if subject $i$ answered yes to question $j$. The Rasch model assumes that the entries $A_{ij}$ are independent Bernoulli random variables with 
\begin{equation}\label{eq-rasch}
    \mathbb{P}_{\beta,\gamma}(A_{ij}=1) = \frac{\exp(\beta_i - \gamma_j)}{1+\exp(\beta_i-\gamma_j)}.
\end{equation}
The model has $I+J-1$ free parameters. When the survey questions are tests, the parameter $\beta_i$ is interpreted as subject $i$'s ability and $\gamma_j$ is interpreted as question $j$'s difficulty. This model is an exponential family where the sufficient statistics are the row and column sums of $A$. 

Suppose that we have a binary matrix $X_0 \in \{0,1\}^{I \times J}$, and we want to know if \eqref{eq-rasch} is an appropriate model for $X_0$. That is, we want to test the composite hypothesis $X_0 \sim \mathbb{P}_{\beta,\gamma}$ for some unspecified $(\beta,\gamma)$. The row and column sums are a sufficient statistic in the Rasch model. Thus, as discussed in Section~\ref{sec:composite}, we can reduce the composite null to a simple null by conditioning. Specifically, let $R_0$ and $C_0$ be the row and column sums of $X_0$ and define $\mathcal{X}$ to be the set of all binary matrices with row sums $R_0$ and column sums $C_0$. If $X_0 \sim \mathbb{P}_{\beta,\gamma}$ for any $(\beta,\gamma)$, then $X_0\mid R_0,C_0 \sim \pi$ where $\pi$ is the uniform distribution on $\mathcal{X}$. Therefore, testing $X_0 \mid R_0,C_0 \sim \pi$ is a conditional test of the composite hypothesis $X_0 \sim \mathbb{P}_{\beta,\gamma}$ for some $(\beta,\gamma)$.

For even a moderate number of items and subjects, sampling directly from the uniform distribution on $\mathcal{X}$ is intractable. A Markov chain with stationary distribution $\pi$ is given in \cite{BC} and a sequential importance sampling approach is described in \cite{SIS}. More recently, a new Markov chain called the rectangle loop algorithm was defined in \cite{Wang}. The rectangle loop algorithm has $\pi$ as its only stationary distribution and was shown to sample more efficiently than the Markov chain in \cite{BC}. Here we use the rectangle loop algorithm with the permuted serial method. 

We used the data set \verb|verbal| from the R package \verb|difR| \cite{difR}. The data set consists of $I=316$ subjects, each of whom answered $J=24$ questions about verbal aggression. Each question involved reading the description of a frustrating situation. The subjects then answered whether they would respond in a certain way such as swearing. The data set also contains each subject's gender. The test statistic used was Andersen's likelihood ratio test \cite{andersen}. This statistic is based on the conditional likelihood,
\[L(\gamma) = \mathbb{P}_{\gamma}\left(X_0\mid R_0\right). \]
The conditional likelihood does not depend on $\beta$ since the row sums $R_0$ are sufficient for $\beta$. Andersen's likelihood ratio equal to the log conditional likelihood from fitting a separate Rasch model for each gender minus the log conditional likelihood from fitting a single Rasch model. The statistics were computed using the function \verb|LRtest| from the R package \verb|eRm| \cite{eRm}. The results are summarized in Figure~\ref{fig Rasch} which shows strong evidence against the single Rasch model. 

\begin{figure}[h]
    \centering
    \includegraphics*[width = 0.4\textwidth]{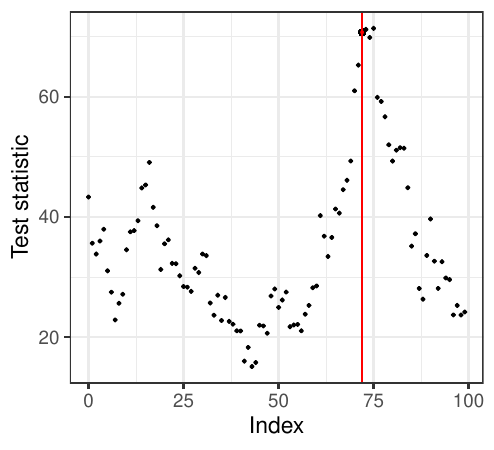}
    \caption{The results from using the serial method to perform a goodness-of-fit test for the Rasch model. The red point equals the value of the test statistic on the observed data $X_0$. The red line is show the random index $m^*$ where we have $Y_{m^*}=X_0$. The black points are equal to the comparison values $T(\widetilde{X}_i)$. To highlight the dependency between the samples $\widetilde{X}_1,\ldots,\widetilde{X}_M$, we have plotted the test statistics in their order prior to being permuted. The test statistic evaluated on the true data, $T(X_0)$ is third largest among $T(X_0), T(\widetilde{X}_1),\ldots, T(\widetilde{X}_M)$. This corresponds to a p-value of $0.03$ giving evidence against the Rasch model.}
    \label{fig Rasch}
\end{figure}

The Rasch model is also used in ecology \cite{SIS}. Binary matrices $A \in \{0,1\}^{I \times J}$ can be used to record the presence and absence of species at different locations. The rows represent different species and the columns represent different locations. A value of $A_{ij}=1$ means that species $i$ is present at location $j$. If $X_0$ is the observed presence/absence matrix for a group of species and locations, then conditioning on the row and column sums of $X_0$ is one way of controlling for species abundance and location habitability. The hypothesis that $X_0$ is drawn from the uniform distribution $\pi$ implies that there are no interactions between species. Statistics such as the number of instances of two species living at the same location can be used to measure the interactions between species. MCMC significance tests can then be used to calibrate these test statistics by producing valid p-values.

\subsection{Political bias in district maps}\label{politics}

In \cite{Gerry1} and \cite{Gerry2}, the authors use MCMC methods to test whether a district map has a political bias. Here a district is defined as collection of census blocks. The number of districts is considered fixed, and a district map is an assignment of each census block to one of the districts. A district map is \textit{valid} if each district is approximately equal in population and the districts are what the authors call ``compact'' which means each district is not too irregular.

The authors of \cite{Gerry1} analyzed the district map of Pennsylvania in 2012. In 2012, Pennsylvania contained 18 districts and roughly 9,000 census blocks. Using our notation, $\mathcal{X}$ is equal to the set of all ways of assigning these 9,000 census blocks to the 18 districts such that the resulting district map is valid. Our observed data $X_0$ is the district map used in 2012 election and the null distribution, $\pi$, is the uniform distribution on $\mathcal{X}$. 

Sampling directly from $\pi$ is impossible. The authors instead construct a Markov chain that uniformly chooses a census block on the border of two districts. The Markov chain then proposes swapping the assignment of the chosen census block. If the swapped block results in a valid district map, then the proposal is accepted. Otherwise, the state of the Markov chain does not change.  By adding a Metropolis--Hastings step, the authors produce a Markov chain on space of valid district maps which has $\pi$ as a stationary distribution. The Markov chain is also shown to be reversible, so that $\widehat{K} = K$. Figure~\ref{fig:MC-map} is from \cite{Gerry1} and shows the 2012 district map of Pennsylvania and a district map produced by taking $2^{40}$ Markov chain steps.

\begin{figure}[t]
    \centering
    \includegraphics[width = 0.4\textwidth]{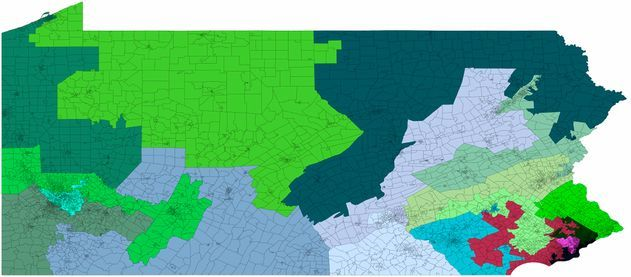}
    \includegraphics[width = 0.4\textwidth]{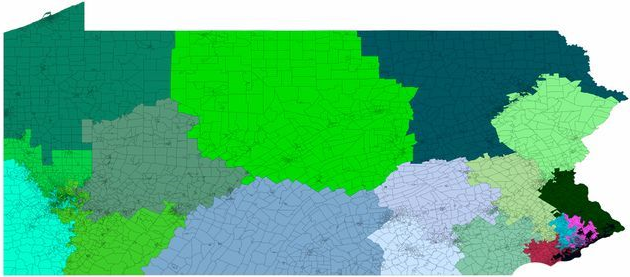}
    \caption{(Left) The 2012 district map of Pennsylvania. (Right) The district map produced by running a Markov chain for $2^{40}$ steps. Figure from \cite{Gerry1}.}
    \label{fig:MC-map}
\end{figure}

In \cite{Gerry2} the authors use a method called the $\sqrt{\varepsilon}$-test to test the hypothesis that the current district map is drawn from the uniform distribution $\pi$. This method is based on the following bound.
\begin{theorem}[Theorem 1.1 in \cite{Gerry1}]\label{prop-sqrt-e}
    Let $K$ be a reversible transition kernel with stationary distribution $\pi$. Suppose that $X_0,\widetilde{X}_1,\ldots,\widetilde{X}_M$ is a Markov chain initialized at $X_0$ with transition kernel $K^L$. Then $\sqrt{2p_{MC}}$ is a valid p-value under the hypothesis $X_0 \sim \pi$.
\end{theorem}
Theorem~\ref{prop-sqrt-e} offers an alternative to exchangeable samplers. Instead of generating exchangeable samples, one can directly use the sequential samples from running a single Markov chain. The $\sqrt{\varepsilon}$-test is thus simpler and easier to  interpret than the previous tests based on exchangeability. However, since $p_{MC} \in [0,1]$, we have $\sqrt{2p_{MC}} > p_{MC}$. Thus, the $\sqrt{\varepsilon}$-test will typically have less power than the methods in Section~\ref{sec:E MCMC}.

Using either the $\sqrt{\varepsilon}$-test or an exchangeable sampler requires choosing a test statistic $T$. The test statistic would ideally give tests that have power against ``political motivated'' alternative distributions. That is, distributions that put high probability on district maps that give a partisan advantage. The authors thus construct test statistics based on voting data from the 2012 Pennsylvania election. For each census block, the authors have the number of votes for the Democratic and Republican parties from the 2012 election. Therefore, given a district map, they can calculate what the election result would have been with that district map. Formally, for each district $j$, there is a function $D_j$ from  $\mathcal{X}$ to the interval $[0,1]$ where for a given district map $x \in \mathcal{X}$, $D_j(x)$ is the proportion of votes for the Democratic Party in district $j$. From here, the authors consider a number of statistics such as the variance of the vector $D(x)$ and the difference between the mean and the median of $D(x)$. The authors find highly significant results for all the statistics they consider. One of the authors testified in a state supreme court case that found the 2012 Pennsylvania district map to be unconstitutional \cite{Gerry1}.

\subsection{Goodness-of-fit testing}\label{sec:gof}

In section~\ref{rasch} we tested the composite null that our observed table was distributed according to a Rasch model with unknown parameters. This was done by conditioning on a sufficient statistic to turn the composite hypothesis into a simple hypothesis. A similar procedure can be used when the null hypothesis is a discrete exponential family. Suppose $\mathcal{X}$ is a finite set, and we have a family of distributions $\mathbb{P}_\eta$ given by,
\begin{equation} \label{eq:exp-fam}\mathbb{P}_\eta(x) = \exp\left(\sum_{j=1}^d S_j(x)\eta_j - Z(\eta)\right), \end{equation}
where $S_j : \mathcal{X} \to \mathbb{R}$ are functions, $\eta \in \Omega \subseteq \mathbb{R}^d$ are the natural parameters and $Z(\eta)$  is a log normalizing constant. The function $S_j$ are sufficient statistics for $\eta$ and if $X \sim \mathbb{P}_\eta$, then $X \mid (S_j(X))_{j=1}^d =s$ is uniform on the set 
\[
    \mathcal{X}_s = \{x \in \mathcal{X} : S_j(x)=s_j \text{ for } 1 \le j \le d\}.
\]
Testing $X_0 \sim \mathbb{P}_\eta$ for some $\eta$ has thus been reduced to testing $X_0 \mid (S_j(X_0))_{j=1}^d=s\sim \pi$ where $\pi$ is the uniform distribution on $\mathcal{X}_s$. In \cite{algebraic}, a general construction is given for finding connected Markov chains on $\mathcal{X}_s$ which converge to a given distribution. These chains can be used in an exchangeable sampler to perform goodness-of-fit tests for any parametric family of the form \eqref{eq:exp-fam}.

In a different direction, \cite{GoF} propose a goodness-of-fit test based on \emph{approximate sufficiency}. There, the authors also wish to test a parametric model $X_0 \sim \mathbb{P}_\eta$ for some $\eta \in \Omega \subseteq \mathbb{R}^d$. They use an approach related to the parametric bootstrap \cite{EfroTibs93}. Using the parametric bootstrap one would sample $\widetilde{X}_1,\ldots,\widetilde{X}_M \stackrel{\text{iid}}{\sim} \mathbb{P}_{\hat{\eta}}$ where $\hat{\eta} = \hat{\eta}(X_0)$ is the maximum likelihood estimate (MLE) given the data $X_0$. The problem with this is that $X_0,\widetilde{X}_1,\ldots,\widetilde{X}_M$ are not exchangeable under the null. To create a valid method, the authors sample $\widetilde{X}_1,\ldots,\widetilde{X}_M$ from the conditional distribution
\begin{equation}\label{eq-cond-mle}\mathbb{P}_{\hat{\eta}'(X_0)}(X \mid \hat{\eta}'(X)=\hat{\eta}'(X_0)), \end{equation}
where $\hat{\eta}'(X)$ is a maximizer of 
\[l(\eta) = \log \mathbb{P}_\eta(X) + \eta^\top U, \quad U \sim \mathcal{N}(0, \tau^2I_d).\]
That is, $\hat{\eta}'(X)$ is a noisy version of the MLE $\hat{\eta}(X)$. This noisy MLE, $\hat{\eta}'(X)$ is used instead of the true MLE $\hat{\eta}(X)$ because in some models, the conditional distribution $X \mid \hat{\eta}(X)$ is degenerate. By adding noise to $\hat{\eta}(X)$ we are essentially conditioning on less information about $X$. The conditional distribution $X \mid \hat{\eta}'(X)$ is thus more variable than $X \mid \hat{\eta}(X)$. 

In \cite{GoF}, the authors compute the conditional distribution \eqref{eq-cond-mle} up to a normalizing constant and derive a Markov chain with \eqref{eq-cond-mle} as a stationary distribution. They also show that if $X_0 \sim \mathbb{P}_{\eta_0}$ for some $\eta_0 \in \Omega$, then 
\[\mathbb{P}_{\hat{\eta}'(X_0)}(X | \hat{\eta}'(X)=\hat{\eta}'(X_0)) \approx \mathbb{P}_{\eta_0}(X | \hat{\eta}'(X)=\hat{\eta}'(X_0)). \]
This implies that if $\widetilde{X}_1,\ldots,\widetilde{X}_M$ are sampled from an exchangeable sampler with stationary distribution \eqref{eq-cond-mle}, then $X_0,\widetilde{X}_1,\ldots,\widetilde{X}_M$ are approximately exchangeable under the null. The authors make this approximation precise and show that their method produces asymptotically valid p-values for goodness-of-fit testing in a variety of parametric models.

\subsection{The conditional permutation test}

In \cite{CPT}, the authors develop a test of conditional independence for ``semi-supervised'' statistical problems. These are statistical problem where we have two data sets, a small labeled data set $(X_j,Y_j,Z_j)_{j=1}^{n}$, and a large unlabeled data set $(X_j', Z_j')_{j=1}^{N}$. In \cite{CPT}, it is assumed that there is enough unlabeled data so that the conditional distribution of $X\mid Z = z$ can be estimated precisely and treated as known. The goal is to test the hypothesis that $Y$ is conditionally independent of $X$ given $Z$. That is, the null hypothesis is
\[(X_j,Y_j,Z_j) \stackrel{\mathrm{iid}}{\sim} P, \]
where $P$ is any distribution that factors as
\begin{equation}\label{eq-cond-ind}P(x,y,z) = Q_{X|Z}(x\mid z)Q_{Y|Z}(y \mid z)Q_Z(z). \end{equation}
The factorization \eqref{eq-cond-ind} is equivalent to the statement that $X$ and $Y$ are conditionally independent given $Z$. As stated above, it is assumed that $Q_{X|Z}$ is known, but no assumptions are made on $Q_{Y|Z}$ or $Q_Z$. These are the same assumptions made of the \emph{conditional randomization test (CRT)}  \cite{CRT}. The CRT is a standard Monte Carlo test that conditions on $(Y_j,Z_j)_{j=1}^n$ and resamples $X_j$. More precisely, let $D_0 = (X_j,Y_j,Z_j)_{j=1}^n$ be our observed data set. For $i=1,\ldots,M$ and $j=1,\ldots,n$, independently sample
\begin{equation}\label{eq-resample} \widetilde{X}_j^{(i)} \sim Q_{X|Z}(\cdot \mid Z_j), \end{equation}
 and define $\widetilde{D}_i = (\widetilde{X}_j^{(i)}, Z_j,Y_j)_{j=1}^n$. Under the conditional independence assumption \eqref{eq-cond-ind}, $D_0,\widetilde{D}_1,\ldots,\widetilde{D}_M$ are an i.i.d. sample from the conditional distribution of $D_0$ given $(Z_j,Y_j)_{j=1}^n$.

 The conditional permutation test (CPT) \cite{CPT} is a variant of the above procedure. In addition to conditioning on $(Z_j,Y_j)_{j=1}^n$, the CPT also conditions on the multi-set $\{X_j\}_{j=1}^n$. This means that when sampling in \eqref{eq-resample}, $\widetilde{X}_1^{(i)},\ldots,\widetilde{X}_n^{(i)}$ must be a permutation of the observed data $X_1,\ldots,X_n$. Thus, instead of sampling $\widetilde{X}_1^{(i)},\ldots,\widetilde{X}_n^{(i)}$ we can sample a permutation $\sigma^{(i)}$ and set 
 \[\widetilde{X}_j^{(i)} = X_{\sigma^{(i)}(j)}, \quad \widetilde{D}_i = \left(\widetilde{X}_j^{(i)},Y_j,Z_j\right)_{j=1}^n. \]
 For $\widetilde{D}_i$ to be drawn from the same conditional distribution as $D_0 \mid (Z_j,Y_j)_{j=1}^n,\{X_j\}_{j=1}^n$ we need that the permutation $\sigma^{(i)}$ is drawn from the distribution
 \begin{equation}\label{eq-Psigma}P(\sigma) = \frac{\prod_{j=1}^n Q(X_{\sigma(j)}\mid Z_j)}{\sum_{\sigma'}\prod_{j=1}^n Q(X_{\sigma'(j)}\mid Z_j)}, \end{equation}
 where the sum in the denominator is over all permutations of $\{1,\ldots,n\}$. The computational cost in evaluating such a sum grows quickly with $n$. To avoid this computational cost, the authors construct a Markov chain on permutations of length $n$ with stationary distribution \eqref{eq-Psigma}. They then use the parallel method to generate samples $\widetilde{D}_1,\ldots,\widetilde{D}_M$ that are exchangeable under the null hypothesis that $X$ and $Y$ are conditionally independent given $Z$. The authors find that the CPT has power comparable to that of the CRT and is more robust against misspecification of the conditional distribution $Q_{X|Z}(\cdot\mid z)$

\section{Conclusion}\label{sec:conclusions}

Markov chain Monte Carlo tests and p-values are useful in a variety of problems when direct sampling from the null distribution is impossible. However, unlike standard Monte Carlo p-values, MCMC p-values are not simply approximating the familiar analytic p-value. The use of a Markov chain introduces dependencies between the Monte Carlo samples. These dependencies mean that the outcome of a MCMC test may be very different to the outcome of a standard test.

\section*{Acknowledgements}

The author would like to thank Persi Diaconis for all his support and John Cherian and Timothy Sudijono for detailed comments on an earlier draft. Thank you to the authors of \cite{Gerry1} for the use of their figure. The author was partially supported by the Brad Efron Fellowship.

\section*{Supplementary material}
Code to reproduce the examples and simulations from this paper is available at \url{https://github.com/Michael-Howes/ MCMC-significance-tests/}

\begin{appendix}  
    \section{Consistency of the permuted serial method}\label{appn} 
        Here we prove Proposition~\ref{prop:serial consistent}. First recall that an irreducible Markov chain is \emph{positive recurrent} if the expected return time to any state $x \in \mathcal{X}$ is finite. To prove Proposition~\ref{prop:serial consistent}, we will use the following two theorems.
        \begin{theorem}[Theorem~3.2.6 in \cite{bremaud2020markov}]\label{theorem:positive}
            An irreducible Markov chain is positive recurrent if and only if there exists a stationary distribution. Moreover, the stationary distribution $\pi$ is, when it exists, unique and $\pi(x) > 0$ for all $x$.
        \end{theorem}
        \begin{theorem}[Theorem~3.3.2 in  \cite{bremaud2020markov}]\label{theorem:ergodic}
            Let $(Y_i)_{i=0}^M$ be an irreducible positive recurrent Markov chain with stationary distribution $\pi$. Let $g:\mathcal{X} \to \mathbb{R}$ be any measurable function with
            \[\mathbb{E}_{\pi}[|g(Y)|] < \infty. \]
            Then for any initial distribution $Y_0 \sim \pi'$ as $M \to \infty$
            \[\frac{1}{M}\sum_{i=1}^M g(Y_i) \stackrel{\mathbb{P}}{\to} \mathbb{E}_\pi[g(Y)]. \]
        \end{theorem}
        \begin{proof}[Proof of Proposition~\ref{prop:serial consistent}]
            Let $m^* \in \{0,\ldots,M\}$, $(Y_i)_{i=0}^M$ and $(\widetilde{X}_i)_{i=0}^M$ be as in the permuted serial method. Note that,
            \begin{align*}
                p_{MC} =&\frac{1}{M+1}\left(\sum_{i=1}^MI[T(\widetilde{X}_i) \ge T(X_0)]+1\right)\\
                =&\frac{1}{M+1}\sum_{i=0}^{m^*-1}I[T(Y_i) \ge T(X_0)] + \frac{1}{M+1} \\
                &+ \frac{1}{M+1}\sum_{i=m^*+1}^M I[T(Y_i) \ge T(X_0)]
            \end{align*}
            Since $\frac{1}{M+1}\to 0$ and $\frac{M}{M+1}\to 1$ as $M\to \infty$, it suffices to show the consistency of
            \begin{align}\label{eq:new_P}
                p_{MC}'=&\frac{1}{M}\sum_{i=0}^{m^*-1}I[T(Y_i) \ge T(X_0)]\\
                &+\frac{1}{M}\sum_{i=m^*+1}^{M}I[T(Y_i) \ge T(X_0)].\nonumber
            \end{align}
            We have assumed that $K^L$ is irreducible and that $\pi$ is a stationary distribution for $K$. This implies that $K^L$ and $\hat{K}^L$ are both irreducible and have stationary distribution $\pi$. Theorem~\ref{theorem:positive} thus implies that $K^L$ and $\widehat{K}^L$ are both irreducible positive recurrent. We can therefore apply Theorem~\ref{theorem:ergodic} with $g(y) = I[T(y) \ge T(X_0)]$ to both $K^L$ and $\widehat{K}^L$. If $Y_0=X_0,Y_1,\ldots,Y_m$ is a Markov chain with transition kernel $K^L$ or $\widehat{K}^L$, then, as $m \to \infty$, 
            \begin{align*}\frac{1}{m}\sum_{i=1}^m I[T(Y_i) \ge T(X_0)] &\stackrel{\mathbb{P}}{\to} \pi(\{y : T(y) \ge T(X_0)\})\\
                &  = p_A. 
            \end{align*}
            In \eqref{eq:new_P}, $Y_{m^*-1},\ldots,Y_0$ are generated by running a Markov chain according to $\widehat{K}^L$ from $X_0$ and likewise $Y_{m^*+1},\ldots,Y_M$ are generated by running a Markov chain according to $K^L$ from $X_0$. Suppose now that we fix $\varepsilon > 0$. As $M$ goes to infinity, $m^*$ and $M-m^*$ both go to infinity with probability one. Thus, as $M \to \infty$, we will have that
            \begin{align*}
                p_1&=\frac{1}{m^*}\sum_{i=1}^{m^*-1}I[T(Y_i) \ge T(X_0)], \text{ and }\\
                p_2&=\frac{1}{M-m^*}\sum_{i=m^*+1}^M I[T(Y_i) \ge T(X_0)],
            \end{align*}
            are within $\varepsilon$ of $p_A$ with probability tending to one. Since $p_{MC}'$ is convex combination of $p_1$ and $p_2$, we can conclude that $p_{MC}'$ will also be within $\varepsilon$ of $p_A$ with probability tending  to one.
        \end{proof}
    \end{appendix}

\end{document}